\documentclass[a4paper, 10pt, twosides]{amsart}
\usepackage{graphicx}
\usepackage{enumitem}
\usepackage{dsfont}
\usepackage[usenames, dvipsnames]{color}
\usepackage{todonotes}
\usepackage[colorlinks=true,
		raiselinks=true,
		linkcolor=MidnightBlue,
		citecolor=Mahogany,
		urlcolor=ForestGreen,
		pdfauthor={Anubhab Dasgupta and Atreyee Kundu},
		pdftitle={Scheduling and Control of Networked Systems},
		pdfkeywords={networked systems, scheduling, control, sparsity},
		pdfsubject={Manuscript},
		plainpages=false]{hyperref}

\usepackage{newtxtext}
\usepackage[varg]{newtxmath}
\usepackage{algorithm,algorithmic}

\DeclareFontFamily{T1}{pzc}{}
\DeclareFontShape{T1}{pzc}{m}{it}{<-> [1.2] pzcmi8t}{}
\DeclareMathAlphabet{\mathpzc}{T1}{pzc}{m}{it}

\usepackage{multicol}

\usepackage{mathtools}
\allowdisplaybreaks
\newtheorem{definition}{Definition}
\newtheorem{fact}{Fact}

\newtheorem{prop}{Proposition}
\newtheorem{remark}{Remark}

\newtheorem{prob}{Problem}

\newtheorem{lemma}{Lemma}

\newcommand{\nn}{\nonumber}
\newcommand{\pmat}[1]{\begin{pmatrix}#1\end{pmatrix}}

\renewcommand{\geq}{\geqslant}

\renewcommand{\leq}{\leqslant}

\newcommand{\abs}[1]{\left\lvert{#1}\right\rvert}

\newcommand{\norm}[1]{\left\lVert#1\right\rVert}

\newcommand{\R}{\mathbb{R}}
\newcommand{\N}{\mathbb{N}}

\newcommand{\Sc}{\mathcal{S}}

\renewcommand{\u}{\upsilon}

\renewcommand{\P}{\mathcal{P}}
\newcommand{\Nset}{\mathcal{Q}}

\DeclareMathOperator{\minimize}{minimize}
\DeclareMathOperator{\sbjto}{subject\;to}

\title[Scheduling and Control of Networked Systems]{A Sparsity Approach to Scheduling and Control\\of Networked Systems}
\author{Anubhab Dasgupta and Atreyee Kundu}
\thanks{Anubhab Dasgupta is with the Department of Mechanical Engineering, Indian Institute of Technology Kharagpur, West Bengal - 721302, India, E-mail: anubhab.dasgupta@kgpian.iitkgp.ac.in. Atreyee Kundu is with the Department of Electrical Engineering, Indian Institute of Technology Kharagpur, West Bengal - 721302, India, E-mail: atreyee@ee.iitkgp.ac.in}
\thanks{A part of this work is accepted at the 62nd IEEE Conference on Decision and Control, Singapore, December 2023.}


\date{\today}

\begin{document}

	\begin{abstract}
       We study the design of scheduling logic and control logic for networked control systems (NCSs) where plants communicate with their remotely located controllers over a shared band-limited communication network. Due to a limited capacity of the network, only a subset of the plants can exchange information with their controllers at any instant of time and the remaining plants operate in open-loop. Our key contribution is a new algorithm that co-designs (a) an allocation scheme of the communication network among the plants (scheduling logic) and (b) the control inputs for the  plants accessing the network (control logic) under which given non-zero initial states are steered to zero in a given time horizon for all the plants in the NCS. Sparse optimization is the primary apparatus for our analysis. We also provide sufficient conditions on the plant dynamics, capacity of the communication network and the given time horizon that lead to a numerically tractable implementation of our algorithm. A numerical experiment is presented to demonstrate the proposed results.
	\end{abstract}
    \maketitle
\section{Introduction}
\label{s:intro}
    Networked Control Systems (NCSs) are an integral part of modern day smart home, smart transportation, smart city, remote surgery, platoons of autonomous vehicles, underwater acoustic communication systems, etc. The network resources for these applications are often limited in the sense that the number of plants sharing a communication network is higher than the capacity of the network. The task of efficient allocation of communication resources is commonly referred to as a \emph{scheduling problem} and the corresponding allocation scheme is called a \emph{scheduling logic}. The set of control inputs to be communicated over the network when the plants have access to it is called a \emph{control logic}. In this paper we will study the co-design of scheduling logic and control logic for NCSs.

    The existing classes of scheduling logic can be classified broadly into two categories --- static (also called periodic, fixed or open-loop) and dynamic (also called non-periodic, or closed-loop) scheduling. In case of static scheduling, a finite length allocation scheme of the shared communication network is determined offline and is applied eternally in a periodic manner. These scheduling logics are easier to implement, often near-optimal and guarantee activation of each sensor and actuator \cite{Peters'16}. In case of dynamic scheduling, the allocation of the network is determined based on some information (e.g., states, outputs, access status of sensors and actuators, etc.) about the plant. These scheduling logics are more suitable to adapt to unforeseen runtime faults in the NCS \cite{Peters'16}. 
    
    Under limited bandwidth of the communication network, dynamics of the plants in an NCS switches between closed-loop and open-loop operations based on whether or not a plant has access to the network. As a natural choice, scheduling and control of NCSs have been addressed as an application of switched systems theory \cite[\S 1.1.2]{Liberzon} widely in the literature. A static scheduling logic that preserves stability of all plants in an NCS is characterized using common Lyapunov functions in \cite{Hristu2001} and piecewise Lyapunov-like functions with average dwell time switching in \cite{Lin2005}. A more general case of co-designing a static scheduling logic and control logic is addressed using linear matrix inequalities optimization with average dwell time technique in \cite{Dai2010}. In \cite{Zhang2006} the authors characterize static switching logic that ensures reachability and observability of the plants under limited communication, and design an observer-based feedback controller for this logic. The corresponding techniques are later extended to the case of constant transmission delays in \cite{Hristu2008} and linear quadratic Gaussian (LQG) control problem in \cite{Hristu_Zhang2008}. In \cite{abc} multiple Lyapunov-like functions and graph theory are employed to design a set of stabilizing switching logics that combined together under certain constraints leads to a stabilizing static scheduling logic. In \cite{def} Lie algebraic techniques for stability of switched systems are used to design stability preserving static scheduling logic under jamming attacks in the network. Stability of Markovian jump linear systems is employed to design stabilizing scheduling logics under data losses in the network in \cite{ghi}.
    
    Among other tools, scheduling and control co-design is addressed using combinatorial optimization with periodic control theory in \cite{Rehbinder2004}. Event-triggered scheduling logics that preserve stability of all plants under communication delays are proposed in \cite{Al-Areqi'15}. In \cite{Quevedo2014} the authors propose a mechanism to allocate network resources by finding optimal node that minimizes a certain cost function in every network time instant. The design of dynamic scheduling logic for stability of each plant under both communication uncertainties and computational limitations is studied in \cite{Saha2015}. In \cite{Gatsis2016} a class of distributed control-aware random network access logics for the sensors such that all control loops are stabilizable, is presented. A dynamic scheduling logic based on predictions of both control performance and channel quality at run-time, is proposed in \cite{Ma2019}. 
    
     In this paper we consider an NCS that consists of \(N\) discrete-time linear plants whose feedback loops are closed through a shared communication network. The network has a limited communication capacity and can support at most \(M\:(0<M<N)\) plants at any instant of time. Consequently, at most \(M\) plants in our NCS can exchange information with their controllers at any instant of time and the remaining at least \(N-M\) plants operate in open-loop, i.e., with a zero control input. Allocation of the communication network among the plants is thus required to ensure that good qualitative properties of each plant in the NCS is preserved. In this paper our objective is to design a scheduling logic and a control logic such that for each plant a given non-zero initial state is steered to zero in a given time horizon.
    
    We observe that if for all plants the given initial state can be steered to zero in open-loop, then no allocation of the shared communication network is required to achieve our objective. Similarly, if a closed-loop operation for steering the given initial state to zero is required at most for \(M\) plants, then the network can be allocated to these plants for all the time. Our objective is, however, no longer trivial if more than \(M\) plants require closed-loop operation, i.e., access to the shared network at least once for its initial state to reach zero in a given time. We consider the setting where all \(N\) plants have this requirement and present an algorithm to design a scheduling logic and a control logic under which a given non-zero initial state is steered to zero in a given time for each plant in the NCS.
    
    The key apparatus in our algorithm is a feasibility problem with sparse constraints (i.e., \(\ell_0\)-constraints). A solution to this feasibility problem is a control logic that for each plant steers the given non-zero initial state to zero in the given time horizon while obeying the constraint that not more than a certain number of plants receive a non-zero control input at any instant of time. Once this control logic is fixed, we allow the plants with non-zero control inputs an access to the shared communication network leading to a scheduling logic. The presence of an \(\ell_0\)-constraint in our design of a control logic, however, gives rise to a non-convex optimization problem. We present sufficient conditions on the plant dynamics, the capacity of the communication network and the given time horizon under which the non-convex optimization problem under consideration admits solutions. We also discuss algorithms to construct these solutions numerically. Further, we study how the non-convex optimization problem can be solved by employing convex relaxations, and provide sufficient conditions on the plant dynamics, the capacity of the communication network, and the time horizon such that a solution to the convex relaxation matches the solution of the original non-convex problem. To the best of our knowledge, our algorithm introduces a new tool for the design of scheduling and control logic for NCSs to the literature.  
   
    The remainder of this paper is organized as follows. We formulate the problem under consideration in \S\ref{s:prob_stat}. Our results appear in \S\ref{s:mainres1} and \S\ref{s:mainres2}. A numerical experiment is presented in \S\ref{s:numex}. We conclude in \S\ref{s:concln} with a brief discussion on future research directions. Proofs of our results appear in a consolidated manner in \S\ref{s:proofs}.
    
    {\bf Notation}. \(\R\) is the set of real numbers and \(\N\) is the set of natural numbers. For \(v\in\R^n\), \(\norm{v}_0\) denotes its \(\ell_0\)-norm, i.e., the number of non-zero elements in \(v\), \(\norm{v}_1\) denotes its \(\ell_1\)-norm, i.e., the sum of absolute values of the elements of \(v\), and \(\norm{v}_2\) denotes its standard Euclidean norm. Let \(0_n\) denote an \(n\)-dimensional zero vector. We extend this notation to also represent an \(m\:(> n)\)-dimensional vector with \(m-n\)-many non-zero entries and \(n\)-many zero entries as \(\pmat{a_1\\\vdots\\a_{m-n}\\0_n}\) (resp., \(\pmat{0_n\\a_1\\\vdots\\a_{m-n}}\)).

\section{Problem statement}
\label{s:prob_stat}
     We consider an NCS with \(N\) discrete-time linear plants whose dynamics are given by
    \begin{align}
    \label{e:plants}
        x_i(t+1) = A_ix_i(t)+b_iu_i(t),\:\:x_{i}(0)=x_{i}^{0},\:\:t=0,1,\ldots,T-1,
    \end{align}
    where \(x_{i}(t)\in\R^{d_i}\) and \(u_{i}(t)\in\R\) are the vector of states and scalar control input of the \(i\)-th plant at time \(t\), respectively, \(A_{i}\in\R^{d_i\times d_i}\), \(b_{i}\in\R^{d_i}\) are constants, \(i=1,2,\ldots,N\). The time horizon \(T\in\N\) is given.
    
     The controllers are remotely located and each plant communicates with its controller through a shared communication network. We operate in the setting where the shared network has a limited communication capacity in the following sense: at any time instant, at most \(M\) plants \((0<M<N)\) can access the network. Consequently, at least \(N-M\) plants operate in open loop (i.e., with \(u_i(t)=0\)) at every time instant. We will, however, assume that the communication network is otherwise ideal in the sense that exchange of information between plants and their controllers is not affected by communication uncertainties.

     \begin{remark}
   \label{rem:equiv}
   \rm{
    	In this paper we assume that if a plant \(i\) does not have access to the shared communication network at time \(t\), then it operates with \(u_i(t) = 0\). Open-loop operation of a plant \(i\) at time \(t\) has been defined in various ways in the literature, for example, \(u_i(t)=\overline{u}\), where \(\overline{u}\) is the latest control input received. 
    }
     \end{remark}  
      
    Let \(\Sc\) be the set of all subsets of \(\{1,2,\ldots,N\}\) containing at most \(M\) distinct elements. 
    \begin{definition}
    \label{d:sched_logic}
    \rm{
         A function \(\gamma_T:\{0,1,\ldots,T-1\} \rightarrow \Sc\), that specifies at every time instant \(t = 0,1,\ldots,T-1\), at most \(M\) plants of the NCS which have access to the shared network at that time, is called a \emph{scheduling logic}.
         }
    \end{definition} 
    \begin{remark}
    \label{rem:no_access}
    \rm{
    Note that \(\emptyset\) is an element of \(\Sc\). Indeed, \(\emptyset\) corresponds to a set with no (distinct) elements. \(\gamma_T(t)=\emptyset\) for some \(t\in\{0,1,\ldots,T-1\}\) implies that no plant has access to the shared network at that time \(t\).}
    \end{remark}
    \begin{definition}
    \label{d:cont_logic}
    \rm{
         A function \(\upsilon_T:\{0,1,\ldots,T-1\}\to\R^N\) denoted by \(\upsilon_T(t) = \pmat{u_1(t)\\u_2(t)\\\vdots\\u_N(t)}\), that specifies at every time instant \(t=0,1,\ldots,T-1\), the control inputs \(u_i(t)\in\R\) for the plants \(i=1,2,\ldots,N\), is called a \emph{control logic}. 
        }
    \end{definition}
    \begin{remark}
   \label{rem:equiv1}
   \rm{
    	In view of our definition of an open-loop operation of plants, for a fixed \(\gamma_T(t)\), we need \(u_i(t)=0\) for all \(i\in\{1,2,\ldots,N\}\) such that \(i\notin\gamma_T(t)\). Further, in our setting, a closed-loop operation of a plant \(i\) with \(u_i(t)=0\) is equivalent to an open-loop operation of plant \(i\). Thus, for a fixed \(\upsilon_T(t)\), excluding all plants \(i\in\{1,2,\ldots,N\}\) with \(u_i(t) = 0\) from \(\gamma_T(t)\) is no loss of generality. 
    }
     \end{remark}  
     
     In the sequel we will refer to a pair, \((\gamma_T,\upsilon_T)\), as a \emph{scheduling and control logic} for the NCS under consideration. Our objective is:
    \begin{prob}
    \label{prob:main}
        Given the plant dynamics, \((A_i, b_i)\), \(i=1,2,\ldots,N\), the initial states, \(x_i(0)=\xi_i\neq 0_{d_i}\), \(i=1,2,\ldots,N\), the capacity of the communication network, \(M\), and the time horizon, \(T\), design a scheduling and control logic, \((\gamma_T,\u_T)\), under which \(x_i(T)=0_{d_i}\) for each plant \(i=1,2,\ldots,N\).
    \end{prob}
    
   
    	We will aim for an \emph{offline} solution to Problem \ref{prob:main} in the sense that we will compute \(\gamma_T(t)\) and \(\u_T(t)\), \(t=0,1,\ldots,T-1\), at one go prior to their application to the NCS. 
   We now move on to our results.
\section{Design of scheduling and control logic}
\label{s:mainres1}
     We begin with a set of observations for the setting where \(x_i(0) = \xi_i\) can be steered to \(x_i(T)=0_{d_i}\) in open-loop for some or all plants \(i\in\{1,2,\ldots,N\}\).
     
     \begin{prop}
     \label{prop:mainres0a}
        Suppose that there exist \(\tau_i\in\{1,2,\ldots,T\}\) such that \(A_i^{\tau_i}\xi_i = 0_{d_i}\), \(i=1,2,\ldots,N\). Then for each plant \(i=1,2,\ldots,N\), \(x_i(T) = 0_{d_i}\) under the scheduling and control logic, \((\gamma_T,\u_T)\), with \(\gamma_T(t) = \emptyset\) and \(\u_T(t) = 0\in\R^N\) for all \(t=0,1,\ldots,T-1\).
     \end{prop}
     
     Proposition \ref{prop:mainres0a} asserts that if for each plant \(i=1,2,\ldots,N\), \(x_i(0)=\xi_i\) can be steered to \(x_i(\tau_i) = 0_{d_i}\), \(\tau_i\leq T\), in open-loop, then our objective is achieved without the plants communicating with their remotely located controllers. This assertion is immediate. We present a proof of Proposition \ref{prop:mainres0a} in \S \ref{s:proofs} for completeness.
     
     Let 
     \begin{align*}
        U_{\xi_i} = \Bigl\{\bigl(u_i(t)\bigr)_{t=0}^{T-1}\:\:\big|\:\:x_i(0)=\xi_i\:\:\text{is steered to}\:\:x_i(T)=0_{d_i}\:\:\text{under}\:u_i(t),\:t=0,1,\ldots,T-1\Bigr\},
    \end{align*}
    \(i=1,2,\ldots,N\). 
    \begin{remark}
    \label{rem:zero_i/p}
    \rm{
        Notice that the set \(U_{\xi_i}\) may include sequences for which not all elements are non-zero. In the light of Remark \ref{rem:equiv1} such values of \(u_i(t)\), \(t\in\{0,1,\ldots,T-1\}\), \(i\in\{1,2,\ldots,N\}\) also relates to the corresponding plant \(i\) not accessing the shared communication network at time \(t\). We will utilize this feature in our results. 
        }
    \end{remark} 
    \begin{prop}
    \label{prop:mainres0b}
        Suppose that the following conditions hold:
        \begin{enumerate}[label=\alph*), leftmargin=*]
            \item There exists \(\mathcal{N}\subseteq\{1,2,\ldots,N\}\) such that
            \begin{enumerate}[label = \roman*), leftmargin=*]
                \item \(\abs{\mathcal{N}}\leq M\),
                \item for each \(i\in\mathcal{N}\), \(A_i^{\tau}\xi_i\neq 0_{d_i}\) for all \(\tau\in\{1,2,\ldots,T\}\), and
                \item for each \(i\in\mathcal{N}\), \(U_{\xi_i}\neq \emptyset\).
            \end{enumerate}
            \item For each \(i\in\{1,2,\ldots,N\}\setminus\mathcal{N}\), there exists \(\tau_i\in\{1,2,\ldots,T\}\) such that \(A_i^{\tau_i}\xi_i=0_{d_i}\). 
        \end{enumerate}
        Then for each plant \(i=1,2,\ldots,N\), \(x_i(T)=0_{d_i}\) under the scheduling and control logic, \((\gamma_T,\u_T)\), with
        \begin{align*}
            \gamma_T(t)&\subseteq\mathcal{N},\:\:t=0,1,\ldots,T-1,\:\text{and}\\
            \nu_T(t) &= \pmat{u_1(t)\\u_2(t)\\\vdots\\u_N(t)}\:\:\text{with}\:u_i(t)=0,\:\text{if}\:i\notin\mathcal{N}, \:u_i(t) = \overline{u}_i(t),\:\text{if}\:i\in\mathcal{N},\:t=0,1,\ldots,T-1,
        \end{align*}
        where \(\biggl(\overline{u}_i(t)\biggr)_{t=0}^{T-1}\in U_{\xi_i}\), \(i\in\mathcal{N}\).
    \end{prop}
    
    Proposition \ref{prop:mainres0b} asserts that if there are at most \(M\) plants for which \(x_i(0)=\xi_i\) cannot be steered to \(x_i(\tau_i)=0_{d_i}\), \(\tau_i\leq T\), in open-loop but can be steered to \(x_i(\overline{\tau}_i)=0_{d_i}\), \(\overline{\tau}_i\leq T\), in closed-loop, and for all the remaining plants \(x_i(0)=\xi_i\) can be steered to \(x_i(\tilde{\tau}_i) = 0_{d_i}\), \(\tilde{\tau}_i\leq T\), in open-loop, then there exists a scheduling and control logic, \((\gamma_T,\u_T)\), under which \(x_i(T) = 0_{d_i}\) for each plant \(i=1,2,\ldots,N\). In particular, the control logic, \(\u_T\), takes the values of \(\bigl(\overline{u}_i(t)\bigr)_{t=0}^{T-1}\in U_{\xi_i}\) for the plants for which closed-loop operation is necessary for taking \(\xi_i\) to \(0_{d_i}\) and value zero for the remaining plants. The scheduling logic, \(\gamma_T\), allows access of the shared communication network to some or all the plants for which closed-loop operation is necessary for taking \(\xi_i\) to \(0_{d_i}\) depending on whether or not the corresponding control input is zero. This correspondence between the scheduling logic, \(\gamma_T\), and the control logic, \(\u_T\), is facilitated by our definition of open-loop operation of the plants, see Remark \ref{rem:equiv1}. We present a proof of Proposition \ref{prop:mainres0b} in \S \ref{s:proofs}.
    
    We now move to a more generalized setting and assume that for each plant \(i=1,2,\ldots,N\), at least one non-zero control input is required to steer \(\xi_i\) to \(0_{d_i}\) in \(T\) time units. This leads to our key results. 
    
    We note that
    \begin{prop}
    \label{prop:mainres0c}
        Suppose that for each plant \(i=1,2,\ldots,N\), \(A_i^\tau\xi\neq 0_{d_i}\) for any \(\tau\in\{1,2,\ldots,T\}\) and \(U_{\xi_i}\neq \emptyset\). If there exists a scheduling and control logic, \((\gamma_T,\nu_T)\), under which \(x_i(T) = 0_{d_i}\) for each plant \(i=1,2,\ldots,N\), then \(T\geq\bigl\lceil\frac{N}{M}\bigr\rceil\).
    \end{prop}
    
    Proposition \ref{prop:mainres0c} deals with the scenario where for each plant \(i=1,2,\ldots,N\), \(x_i(0)=\xi_i\) can be steered to \(x_i(T) = 0_{d_i}\) in closed-loop but not in open-loop, and provides a necessary condition on the time horizon as a function of the total number of plants and the capacity of the shared communication network for the existence of a scheduling and control logic, \((\gamma_T,\u_T)\), that achieves the desired objective. A proof of Proposition \ref{prop:mainres0c} is presented in \S \ref{s:proofs}. In the sequel we will operate under the assumption that the given time horizon, \(T\), is at least \(\bigl\lceil\frac{N}{M}\bigr\rceil\). It is easy to see that \(T\geq\bigl\lceil\frac{N}{M}\bigr\rceil\) is, however, not a sufficient condition for the existence of a scheduling and control logic, \((\gamma_T,\u_T)\), under which \(x_i(T)=0_{d_i}\) for each plant \(i=1,2,\ldots,N\). Indeed, consider \(N=2\), \(M=1\) and \(T = 3>2 = \bigl\lceil\frac{N}{M}\rceil\). Suppose that \(U_{\xi_1}=\{(2,2,2)\}\) and \(U_{\xi_2}=\{(3,3,3)\}\). In order to achieve \(x_i(3) = 0_{d_i}\), \(i=1,2\), an access to the shared communication network at each time \(t=0,1,2\) is required for both the plants which is not supported by the communication capacity, \(M=1\).

      We present an algorithm (Algorithm \ref{algo:sched-con_design}) that designs a scheduling and control logic, \((\gamma_T,\upsilon_T)\), when all plants \(i=1,2,\ldots,N\) require at least one non-zero control input to steer \(\xi_i\) to \(0_{d_i}\) in \(T\) time units. Our algorithm involves two steps: 
      \begin{itemize}[label=\(\circ\),leftmargin=*]
        \item First, it solves the feasibility problem \eqref{e:feasprob1} to compute the control logic, \(\upsilon_T\). A solution \(\bigl(\upsilon_T(t)\bigr)_{t=0}^{T-1}\) to \eqref{e:feasprob1} has the following properties: 
            \begin{enumerate}[label = (\alph*),leftmargin=*]
                \item for each plant \(i=1,2,\ldots,N\), \(\bigl(u_i(t)\bigr)_{t=0}^{T-1}\) steers the given initial state \(\xi_i\) to \(0_{d_i}\) in (at most) \(T\) units of time, and 
                \item at every time \(t=0,1,\ldots,T-1\), \(u_i(t)\neq 0\) at most for \(M\)-many plants \(i\in\{1,2,\ldots,N\}\).
            \end{enumerate}
        \item Second, at each time \(t=0,1,\ldots,T-1\), the plants \(i\in\{1,2,\ldots,N\}\) with \(u_i(t)\neq 0\) are assigned to \(\gamma_T(t)\), i.e., they are allowed an access to the shared communication network. These plants receive their non-zero control inputs at time \(t\) and the plants \(i\notin\gamma_T(t)\) operate in open-loop with \(u_i(t) = 0\), \(t=0,1,\ldots,T-1\).
      \end{itemize} 
      Notice that \eqref{e:feasprob1} is a sparse optimization problem in the sense that it involves an upper bound on the number of non-zero elements in \(\upsilon_T\) (and thus a lower bound on the number of zero elements in \(\upsilon_T\)) as an \(\ell_0\)-constraint. This sparse constraint accommodates the capacity of the shared communication network, \(M\), in our design of control logic, \(\upsilon_T\), and thus of scheduling logic, \(\gamma_T\). 
    \begin{algorithm}[htbp]
			\caption{Construction of a scheduling and control logic, \((\gamma_T,\upsilon_T)\)} \label{algo:sched-con_design}
		\begin{algorithmic}[1]
			\renewcommand{\algorithmicrequire}{\textbf{Input:}}
			\renewcommand{\algorithmicensure}{\textbf{Output:}}
			
			\REQUIRE The plant dynamics, \((A_i,b_i)\), \(i=1,2,\ldots,N\), the capacity of the communication network, \(M\), the initial states, \(x_i(0) = \xi_i\), \(i=1,2,\ldots,N\), and time horizon, \(T\).
			\ENSURE A scheduling and control logic, \((\gamma_T,\upsilon_T)\).
			
			\STATE Solve the following feasibility problem for \(\upsilon_T\):
			\begin{align}
		\label{e:feasprob1}
			\underset{\upsilon_T(0),\ldots,\upsilon_T(T-1)}\minimize&\quad1\\
			\sbjto&\:
			\begin{cases}
				x_i(t+1) = A_i x_i(t) + b_i {u}_i(t),\\
				\quad\quad\:\: t=0,1,\ldots,T-1,\:\: i=1,2,\ldots,N,\\
				x_i(0) = \xi_i,\:x_i(T) = 0_{d_i},
				\:\: i=1,2,\ldots,N,\\
				\norm{\upsilon_T(t)}_0\leq M,\:t=0,1,\ldots,T-1.\nn
			\end{cases}
		\end{align}
		
			\STATE If there exists a solution \(\upsilon_T\) to the feasibility problem \eqref{e:feasprob1}, then go to Step \ref{step:new1}. Otherwise, terminate the algorithm.
		
			\STATE\label{step:new1} Construct \(\gamma_T(t)\) as the set containing the elements of the set \(\{1,2,\ldots,N\}\) such that \({u}_i(t)\neq 0\), (i.e., the \(i\)-th component of \(\upsilon_T(t)\) is non-zero), \(t=0,1,\ldots,T-1\). 
			
		\end{algorithmic}
	\end{algorithm}
		
	The following proposition asserts that for each plant \(i=1,2,\ldots,N\), the state \(x_i(0)=\xi_i\) is steered to \(x_i(T)=0_{d_i}\) under a scheduling and control logic, \((\gamma_T,\upsilon_T)\), obtained from our algorithm. 
    
    \begin{prop}
    \label{prop:mainres1}
        Suppose that the plant dynamics, \((A_i,b_i)\), \(i=1,2,\ldots,N\), the capacity of the communication network, \(M\), the initial states, \(x_i(0)=\xi_i\), \(i=1,2,\ldots,N\), and the time horizon, \(T\), are given. Let \((\gamma_T,\upsilon_T)\) be a scheduling and control logic obtained from Algorithm \ref{algo:sched-con_design}. Then for each plant \(i=1,2,\ldots,N\), \(x_i(T) = 0_{d_i}\) under \((\gamma_T,\upsilon_T)\).
    \end{prop}
    
    A proof of Proposition \ref{prop:mainres1} is provided in \S \ref{s:proofs}.
    
    \begin{remark}
    \label{rem:compa1}
    \rm{
        A vast body of the existing literature on scheduling and control of networked systems first fixes a scheduling logic and then designs the control inputs such that all plants achieve a desired property under the scheduling logic in consideration. In contrast, we first construct the sequences of control inputs that lead to a desired property for all the plants while obeying the limitation in communication capacity, and then construct the scheduling logic by employing non-zero instances of the control inputs.
        }
  \end{remark}
  
     \begin{remark}
    \label{rem:compa3}
    \rm{
        A scheduling and control logic, \((\gamma_T,\upsilon_T)\), obtained from Algorithm \ref{algo:sched-con_design} is \emph{static} or \emph{open-loop} in the sense that \(\gamma_T(t)\) and \(\upsilon_T(t)\) are designed offline at one go for all \(t=0,1,\ldots,T-1\). While open-loop scheduling and control logic are often easier to implement, they do not adapt to unforeseen faults in the plants and/or other components in the NCS that may occur during the operation.
    }
    \end{remark}

    \begin{remark}
    \label{rem:compa2}
    \rm{
        In the recent work \cite{Ikeda2022}, the authors presented a sparsity based approach (joint \(\ell_0\) and \(L^0\)-optimization) for scheduling of a continuous-time linear networked system towards achieving a certain controllability metric. While sparse optimization is also the key apparatus in our work, the objective differs from \cite{Ikeda2022} in the following ways: we use sparse optimization (\(\ell_0\)-optimization) for the design of scheduling and control logic for NCSs with multiple discrete-time plants operating under a limited communication capacity with the objective to steer non-zero initial states for each plant to zero in a given time horizon.
        }
    \end{remark}
    
   This concludes our discussion on scheduling and control logic, \((\gamma_T,\upsilon_T)\), under which \(x_i(0)=\xi_i\) is steered to \(x_i(T) = 0_{d_i}\) for each plant \(i=1,2,\ldots,N\). In the following section we address numerical implementation issues related to Algorithm \ref{algo:sched-con_design}. 
   
\section{Numerical implementation of Algorithm \ref{algo:sched-con_design}}
\label{s:mainres2}   
   We observe that
    \begin{fact}
    \label{fact:non-convexity}
        The feasibility problem \eqref{e:feasprob1} is not convex.
    \end{fact}
    Notice that since \(\ell_0\)-norm is not a convex function \cite[Chapter 2]{Nagahara_book}, all constraints in \eqref{e:feasprob1} are not convex and the definition of a convex optimization problem is violated. Fact \ref{fact:non-convexity} is immediate. Consequently, computing a solution to \eqref{e:feasprob1} numerically is a difficult task. We describe our results towards achieving a solution to \eqref{e:feasprob1} in this section.
    
    We note that the existence of a solution to \eqref{e:feasprob1} is equivalent to the following:
	\begin{prop}
	\label{prop:mainres1a}
		Suppose that \(U_{\xi_i}\neq\emptyset\) for all \(i=1,2,\ldots,N\) and there exist elements \(\bigl(\overline{u}_i(t)\bigr)_{t=0}^{T-1}\in U_{\xi_i}\), \(i=1,2,\ldots,N\) such that \(\norm{\pmat{\overline{u}_1(t)\\\overline{u}_2(t)\\\vdots\\\overline{u}_N(t)}}_{0}\leq M\) for all \(t=0,1,\ldots,T-1\). Then there exists a solution to the feasibility problem \eqref{e:feasprob1}.
	\end{prop}

	Proposition \ref{prop:mainres1a} states that if there exist \(\bigl(\overline{u}_i(t)\bigr)_{t=0}^{T-1}\in U_{\xi_i}\), \(i=1,2,\ldots,N\) such that for each \(t=0,1,\ldots,T-1\), \(\overline{u}_i(t)\neq 0\) for at most \(M\)-many \(i\in\{1,2,\ldots,N\}\), then the feasibility problem \eqref{e:feasprob1} admits a solution. This assertion is easy to see. In particular, a solution to \eqref{e:feasprob1} is
	\(\upsilon_T(t) = \pmat{u_1(t)\\u_2(t)\\\vdots\\u_N(t)}= \pmat{\overline{u}_1(t)\\\overline{u}_2(t)\\\vdots\\\overline{u}_N(t)}\), \(t=0,1,\ldots,T-1\). An exhaustive search over all elements of \(U_{\xi_i}\), \(i=1,2,\ldots,N\) allows us to find suitable combinations of \(\biggl(\bigl(\overline{u}_i(t)\bigr)_{t=0}^{T-1}\biggr)_{i=1}^{N}\), if exists. 

We now present a set of sufficient conditions on the plant dynamics, \((A_i,b_i)\), \(i=1,2,\ldots,N\), the capacity of the communication network, \(M\), and the time horizon, \(T\), such that \eqref{e:feasprob1} admits a solution and an exhaustive search to compute the same is not needed. We also discuss algorithmic construction of these solutions.  

Notice that we require 
		\[
            \displaystyle{x_i(T)=A_i^T \xi_i +\Phi_i\pmat{\overline{u}_i(0)\\\overline{u}_i(1)\\\vdots\\\overline{u}_i(T-1)} = 0_{d_i}},
        \]
        where \(\Phi_i= \pmat{A_i^{T-1}b_i & A_i^{T-2}b_i & \ldots & A_ib_i & b_i}\). 
        
         Let \(\Psi_i = \pmat{A_i^{d_i-1}b_i & A_i^{d_i-2}b_i & \cdots & A_ib_i & b_i}\), \(i=1,2,\ldots,N\).
    \begin{definition}
    \label{d:reachability}
    \rm{
        We call a plant \(i\) \emph{reachable} if the pair \((A_i,b_i)\in\R^{d_i\times d_i}\times\R^{d_i}\) satisfies 
            \(\text{rank}(\Psi_i) = d_i\).
            }
    \end{definition}
     
    \begin{prop}
    \label{prop:mainres3}
        Suppose that the following conditions hold:
        \begin{enumerate}[label = C\arabic*), leftmargin = *]
            \item\label{condn:c1} Each plant \(i=1,2,\ldots,N\) is reachable.
            \item\label{condn:c2} There exist \(\hat{d}_j\in\N\), \(\P_j\subseteq\{1,2,\ldots,N\}\), \(j=1,2,\ldots,\bigl\lceil\frac{N}{M}\bigr\rceil\) such that
            \begin{enumerate}[label=\alph*), leftmargin=*]
                \item \(\abs{\P_j}\leq M\), \(j=1,2,\ldots,\bigl\lceil\frac{N}{M}\bigr\rceil\),
                \item \(\P_\ell\cap\P_m = \emptyset\) for all \(\ell,m = 1,2,\ldots,\bigl\lceil\frac{N}{M}\bigr\rceil\), \(\ell\neq m\),
                \item \(\displaystyle{\bigcup_{j=1}^{\bigl\lceil\frac{N}{M}\bigr\rceil}}\P_j = \{1,2,\ldots,N\}\),
                \item \(\hat{d}_j>d_i\) for all \(i\in\P_j\), \(j=1,2,\ldots,\bigl\lceil\frac{N}{M}\bigr\rceil\), and
                \item \(\displaystyle{\sum_{j=1}^{\bigl\lceil\frac{N}{M}\bigr\rceil}}\hat{d}_j \leq T\).
            \end{enumerate}
        \end{enumerate}
        Then a control logic, \(\upsilon_T\), obtained from Algorithm \ref{algo:input_design1} is a solution to the feasibility problem \eqref{e:feasprob1}.
    \end{prop}
    \begin{algorithm*}[htbp]
			\caption{Construction of a control logic, \(\upsilon_T\), when conditions \ref{condn:c1}-\ref{condn:c2} are satisfied} \label{algo:input_design1}
		\begin{algorithmic}[1]
			\renewcommand{\algorithmicrequire}{\textbf{Input:}}
			\renewcommand{\algorithmicensure}{\textbf{Output:}}
			
			\REQUIRE The plant dynamics, \((A_i,b_i)\), \(i=1,2,\ldots,N\), the capacity of the communication network, \(M\), the initial states, \(x_i(0) = \xi_i\), \(i=1,2,\ldots,N\), time horizon, \(T\), the numbers, \(\hat{d}_j\), \(j=1,2,\ldots,\bigl\lceil\frac{N}{M}\bigr\rceil\) and the sets, \(\P_j\), \(j=1,2,\ldots,\bigl\lceil\frac{N}{M}\bigr\rceil\).
			\ENSURE A control logic, \(\upsilon_T\).
			
			\FOR {\(j=1,2,\ldots,\bigl\lceil\frac{N}{M}\bigr\rceil\)}
                \STATE Compute \(\displaystyle{\tilde{d}_j = \sum_{k=1}^{j-1}\hat{d}_k}\).
                \FOR {each \(i\in\P_j\)}
                    \STATE Set \begin{align}
                    			\label{e:u_comp1}
                    		\displaystyle{\pmat{u_i\bigl(\tilde{d}_j+0\bigr)\\\vdots\\u_i\bigl(\tilde{d}_j+\hat{d}_j-1\bigr)} = \pmat{0_{\hat{d}_j-d_i}\\-\Psi_i^{-1}A_i^{\hat{d}_j}A_i^{\tilde{d}_j}{\xi}_i}}
				\end{align} 
                    and \({u}_i(\tau) = 0\) for all \(\tau\in\{0,1,\ldots,T-1\}\setminus\bigl\{\tilde{d}_j+0,\ldots,\tilde{d}_j+\hat{d}_j-1\bigr\}\).
                \ENDFOR
            \ENDFOR
                 \STATE Output \(\upsilon_T(t) = \pmat{u_1(t)\\u_2(t)\\\vdots\\u_N(t)}\), \(t=0,1,\ldots,T-1\).
		\end{algorithmic}
	\end{algorithm*}

    Proposition \ref{prop:mainres3} provides our first set of sufficient conditions under which the existence of a solution to the feasibility problem \eqref{e:feasprob1} is guaranteed. This solution can be computed by employing Algorithm \ref{algo:input_design1}. We split the set of all plants \(\{1,2,\ldots,N\}\) into disjoint subsets, \(\P_j\), \(j=1,2,\ldots,\bigl\lceil\frac{N}{M}\bigr\rceil\) and under reachability assumption use \(\hat{d}_j\), \(j=1,2,\ldots,\bigl\lceil\frac{N}{M}\bigr\rceil\) as the time duration for steering the state of the elements in \(\P_j\) to zero with at most \(d_i\)-many non-zero control inputs. The cardinality of each \(\P_j\) ensures that no more than \(M\) plants have non-zero control inputs at any time instant \(t\) and the upper bound on the sum of \(\hat{d}_j\), \(j=1,2,\ldots,\bigl\lceil\frac{N}{M}\bigr\rceil\) ensures that the given initial states of all plants are steered to zero in the given time horizon \(T\). We present a proof of Proposition \ref{prop:mainres3} in \S \ref{s:proofs}.

    \begin{prop}
    \label{prop:mainres4}
        Suppose that the following conditions hold:
        \begin{enumerate}[label=D\arabic*), leftmargin = *]
            \item\label{condn:d1} Each plant \(i=1,2,\ldots,N\) is reachable.
            \item\label{condn:d2} There exist \(\N\ni\hat{d}_i>d_i\) for all \(i=1,2,\ldots,N\) and \(\Nset_j\subseteq\{1,2,\ldots,N\}\), \(j=1,2,\ldots,p\) with \(p\leq M\) such that
            \begin{enumerate}[label=\alph*), leftmargin =*]
                \item \(\Nset_m\cap\Nset_n=\emptyset\) for all \(m,n=1,2,\ldots,p\), \(m\neq n\),
                \item \(\displaystyle{\bigcup_{j=1}^{p}\Nset_j=\{1,2,\ldots,N\}}\), and
                \item \(\displaystyle{\sum_{i\in\Nset_j}}\hat{d}_i\leq T\) for all \(\Nset_j\), \(j=1,2,\ldots,p\).
            \end{enumerate}
        \end{enumerate}
        Then a control logic, \(\upsilon_T\), obtained from Algorithm \ref{algo:input_design2} is a solution to the feasibility problem \eqref{e:feasprob1}.
    \end{prop}
     \begin{algorithm*}[htbp]
			\caption{Construction of a control logic, \(\upsilon_T\) when conditions \ref{condn:d1}-\ref{condn:d2} are satisfied} \label{algo:input_design2}
		\begin{algorithmic}[1]
			\renewcommand{\algorithmicrequire}{\textbf{Input:}}
			\renewcommand{\algorithmicensure}{\textbf{Output:}}
			
			\REQUIRE The plant dynamics, \((A_i,b_i)\), \(i=1,2,\ldots,N\), the initial states, \(x_i(0) = \xi_i\), \(i=1,2,\ldots,N\), time horizon, \(T\), the integers, \(\hat{d}_i\), \(i=1,2,\ldots,N\) and the sets, \(\Nset_j\), \(j=1,2,\ldots,p\).
			\ENSURE A control logic, \(\upsilon_T\).
			
			Let \(\Nset_j(k)\) denote the \(k\)-th element of the set \(\Nset_j\), \(k=1,2,\ldots,\abs{\Nset_j}\).
			
			\FOR {each \(j=1,2,\ldots,p\)}
                \FOR {each \(k=1,2,\ldots,\abs{\Nset_j}\)}
                    \STATE Compute \(\displaystyle{\tilde{d}_{\Nset_j(k)} = \sum_{\ell=1}^{k-1}\hat{d}_{\Nset_j(\ell)}}\).
                    \STATE Set 
                    \begin{align}
                    \label{e:u_comp2}
                    	\pmat{u_{\Nset_j(k)}\bigl(\tilde{d}_{\Nset_j(k)}+0\bigr)\\\vdots\\u_{\Nset_j(k)}\bigl(\tilde{d}_{\Nset_j(k)}+\hat{d}_{\Nset_j(k)}-1\bigr)} = \pmat{0_{\hat{d}_{\Nset_j(k)}-d_{\Nset_j(k)}}\\-\Psi_{\Nset_j(k)}^{-1}A_{\Nset_j(k)}^{\hat{d}_{\Nset_j(k)}}A_{\Nset_j(k)}^{\tilde{d}_{\Nset_j(k)}}{\xi}_{\Nset_j(k)}}
                    \end{align}
                    and \({u}_{\Nset_j(k)}(\tau) = 0\) for all \(\tau\in\{0,1,\ldots,T-1\}\setminus\bigl\{\tilde{d}_{\Nset_j(k)}+0,\ldots,\tilde{d}_{\Nset_j(k)}+\hat{d}_{\Nset_j(k)}-1\bigr\}\).
                \ENDFOR
            \ENDFOR
                 \STATE Output \(\upsilon_T(t) = \pmat{u_1(t)\\u_2(t)\\\vdots\\u_N(t)}\), \(t=0,1,\ldots,T-1\).
		\end{algorithmic}
	\end{algorithm*}

    The purpose of Proposition \ref{prop:mainres4} is to provide our second set of sufficient conditions under which the feasibility problem \eqref{e:feasprob1} admits a solution. This solution can be computed by employing Algorithm \ref{algo:input_design2}. Under reachability assumption for each plant \(i=1,2,\ldots,N\), we use \(\hat{d}_i\) as the time duration for steering the state of the plant to \(0_{d_i}\) with at most \(d_i\)-many non-zero control inputs. We split the set of all plants \(\{1,2,\ldots,N\}\) into at most \(M\) disjoint subsets, \(\Nset_j\), \(j=1,2,\ldots,p\:(\leq M)\) such that the sum of \(\hat{d}_i\) for all elements \(i\) in any \(\Nset_j\) does not exceed \(T\). The number of \(\Nset_j\), \(j=1,2,\ldots,p\) in use ensures that not more than \(M\) plants have non-zero control inputs at any time instant \(t\) and the upper bound on the sum of \(\hat{d}_i\), \(i\in\Nset_j\), \(j\in\{1,2,\ldots,p\}\) ensures that the given initial states of all plants are steered to zero in the given time horizon \(T\). We present a proof of Proposition \ref{prop:mainres4} in \S \ref{s:proofs}.
    
    At this point the reader may wonder about the connection and comparison between Proposition \ref{prop:mainres3} (Algorithm \ref{algo:input_design1}) and Proposition \ref{prop:mainres4} (Algorithm \ref{algo:input_design2}). Consider a matrix, \(\mathcal{M}\), with \(M\) rows and \(T\) columns. Suppose that we want to fill in the entries of \(\mathcal{M}\) with elements from the set \(\{1,2,\ldots,N\}\) obeying that each element must appear for a pre-specified number of consecutive instances \(\hat{d}_j\), \(j\in\{1,2,\ldots,\bigl\lceil\frac{N}{M}\bigr\rceil\}\) (resp., \(\hat{d}_i\), \(i\in\{1,2,\ldots,N\}\)). Algorithm \ref{algo:input_design1} splits the columns into segments of length \(\hat{d}_1\), \(\hat{d}_2,\ldots\), \(\hat{d}_{\big\lceil\frac{N}{M}\bigr\rceil}\) and fills in the elements from \(\P_1\), \(\P_2,\ldots\), \(\P_{\big\lceil\frac{N}{M}\bigr\rceil}\), respectively, while Algorithm \ref{algo:input_design2} assigns elements from \(\Nset_1\), \(\Nset_2,\ldots\), \(\Nset_p\) for the corresponding \(\hat{d}_i\), \(i=1,2,\ldots,N\) duration of time along the rows \(1,2,\ldots,p\:(\leq M)\), respectively. Propositions \ref{prop:mainres3} and \ref{prop:mainres4} are also algebraically related. If \(\abs{\Nset_j}=\bigl\lceil\frac{N}{M}\bigr\rceil\) for all \(j=1,2,\ldots,p\), then we can pick \(\P_k = \{\Nset_1(k),\Nset_2(k),\ldots,\Nset_p(k)\}\) and \(\hat{d}_k = \max\{\hat{d}_{\Nset_1(k)},\)\\\(\hat{d}_{\Nset_2(k)},\ldots,\hat{d}_{\Nset_p(k)}\}\), \(k=1,2,\ldots,\bigl\lceil\frac{N}{M}\bigr\rceil\). Given a time horizon, \(T\), it is clear that Proposition \ref{prop:mainres4} can cater to a bigger \(N\) compared to Proposition \ref{prop:mainres3} as the choice of \(\hat{d}_i\) is specific to \(i\) and not a maximal value over a set of \(i\)'s which can be of different dimension, \(d_i\). In other words, Proposition \ref{prop:mainres4} is more useful when \(d_i\neq d\) for all \(i=1,2,\ldots,N\). Consider, for example, an NCS with \(N=4\) and \(M=2\). Let \(d_1=1\), \(d_2=2\), \(d_3=3\) and \(d_4=4\). Let \(T=7\). We have that Proposition \ref{prop:mainres4} holds with \(\hat{d}_1=2\), \(\hat{d}_2=3\), \(\hat{d}_3=4\), \(\hat{d}_4=5\) and \(\Nset_1=\{1,4\}\), \(\Nset_2=\{2,3\}\). However, for any choice of disjoint \(\P_1\), \(\P_2\subseteq\{1,2,3,4\}\), we cannot choose \(\hat{d}_1\), \(\hat{d}_2\) such that \(\hat{d}_1+\hat{d}_2\leq 7\).  
		
	\begin{remark}
	\label{rem:zero_entries}
    \rm{
        We note the following feature of a scheduling and control logic, \((\gamma_T,\u_T)\), obtained from Algorithm \ref{algo:sched-con_design}, where the control logic, \(\u_T\), is computed by employing Algorithm \ref{algo:input_design1} (resp., Algorithm \ref{algo:input_design2}): For each plant \(i\in\{1,2,\ldots,N\}\), we have at most \(d_i\)-many non-zero control inputs in a \(\hat{d}_j\), \(j\in\{1,2,\ldots,\bigl\lceil\frac{N}{M}\bigr\rceil\}\) (resp., \(\hat{d}_i\)) duration of time. The remaining at least \(\hat{d}_j-d_i\) (resp., \(\hat{d}_i-d_i\)) elements are \(0\). Consequently, for the scheduling logic, \(\gamma_T\), there exist time instants \(\tau\in\{0,1,\ldots,T-1\}\) such that \(\gamma_T(\tau)=\emptyset\), i.e., no plant has access to the shared communication network. 
    }
	\end{remark}

    Our next set of results is aimed towards arriving at a solution to the feasibility problem \eqref{e:feasprob1} by solving a set of convex optimization problems. We provide sufficient conditions on the plant dynamics, \((A_i,b_i)\), \(i=1,2,\ldots,N\), capacity of the communication network, \(M\) and the time horizon, \(T\), under which this is achievable. Prior to presenting these results, we catalog a few definitions.
    
    \begin{definition}
    \label{d:sparse}
    \rm{
        A vector \(z\in\R^T\) is called \emph{\(s\)-sparse} if it satisfies \(\norm{z}_0\leq s\).
    }
    \end{definition}
    Let \(\Sigma_s\) denote the set of all \(s\)-sparse vectors \(z\in R^{T}\).
    
    \begin{definition}
    \label{d:rip_mat}
    \rm{
        A matrix \(\Gamma\in\R^{p\times T}\) satisfies the \emph{restricted isometry property of order \(s\)} if there exists \(\delta_s\in]0,1[\) such that
        \[
            (1-\delta_s)\norm{z}_2^2\leq\norm{\Gamma z}_2^2\leq (1+\delta_s)\norm{z}_2^2
        \]
        holds for all \(z\in\Sigma_s\).
    }
    \end{definition}
    
    \begin{prop}
    \label{prop:mainres5}
        Suppose that the following conditions hold:
        \begin{enumerate}[label=E\arabic*),leftmargin=*]
            \item\label{condn:e1} Each plant \(i=1,2,\ldots,N\) is reachable.
            \item\label{condn:e2} \(T > d_i\) for each \(i=1,2,\ldots,N\).
            \item\label{condn:e3} For each \(i=1,2,\ldots,N\), the solution, \(\bigl(\overline{u}_i(t)\bigr)_{t=0}^{T-1}\), to the optimization problem
                \begin{align}
                \label{e:optprob1}
                    \underset{\bigl(u_i(t)\bigr)_{t=0}^{T-1}\in U_{\xi_i}}\minimize&\:\norm{\pmat{u_i(0)\\u_i(1)\\\vdots\\u_i(T-1)}}_{0}
                \end{align}
                is \(s_i\)-sparse.
            \item\label{condn:e4} There exist \(\mathcal{R}_j\subset\{1,2,\ldots,N\}\), \(j=1,2,\ldots,M\) satisfying
            \begin{enumerate}[label=\alph*),leftmargin = *]
                \item \(\mathcal{R}_m\cap\mathcal{R}_n=\emptyset\) for all \(m,n=1,2,\ldots,M\), \(m\neq n\),
                \item \(\displaystyle{\bigcup_{j=1}^{M}\mathcal{R}_j}=\{1,2,\ldots,N\}\),
                \item for each \(j=1,2,\ldots,M\), \(\displaystyle{\sum_{i\in\mathcal{R}_j}s_i\leq T}\).
            \end{enumerate}
        \end{enumerate}
        Then the control logic, \(\upsilon_T(t) = \pmat{\overline{u}_1(t)\\\overline{u}_2(t)\\\vdots\\\overline{u}_N(t)}\), \(t=0,1,\ldots,T-1\)
        is a solution to the feasibility problem \eqref{e:feasprob1}.
    \end{prop}
    
    Proposition \ref{prop:mainres5} gives sufficient conditions for the existence of a solution to the feasibility problem \eqref{e:feasprob1} in terms of solutions to the optimization problem \eqref{e:optprob1} that satisfy certain conditions. We work under the assumption that each plant \(i=1,2,\ldots,N\) is reachable and the time horizon \(T>d_i\) for each plant \(i=1,2,\ldots,N\). If solutions \(\biggl(\bigl(\overline{u}_i(t)\bigr)_{t=0}^{T-1}\biggr)_{i=1}^{N}\) to the optimization problem \eqref{e:optprob1} are \(s_i\)-sparse, \(i=1,2,\ldots,N\) and all plants \(i=1,2,\ldots,N\) can be combined into \(M\) disjoint sets, \(\mathcal{R}_j\subseteq\{1,2,\ldots,N\}\) such that i) for each \(j=1,2,\ldots,M\) and \(t=0,1,\ldots,T-1\), at most one element in each \(\mathcal{R}_j\) has a non-zero value of \(\overline{u}_i(t)\) and ii) for each \(j=1,2,\ldots,M\), the sum of the number of non-zero elements in \(\bigl(\overline{u}_i(t)\bigr)_{t=0}^{T-1}\) over all \(i\in\mathcal{R}_j\) does not exceed the given time horizon, \(T\), then the control logic, \(\upsilon_T(t) = \pmat{\overline{u}_1(t)\\\overline{u}_2(t)\\\vdots\\\overline{u}_N(t)}\), \(t=0,1,\ldots,T-1\), is a solution to the feasibility problem \eqref{e:feasprob1}. We present a proof of Proposition \ref{prop:mainres5} in \S\ref{s:proofs}. Notice that \eqref{e:optprob1} is a non-convex optimization problem as its objective function is non-convex. 
    
    The following result asserts that under certain conditions on the plant dynamics, \((A_i,b_i)\), \(i=1,2,\ldots,N\), capacity of the communication network, \(M\) and the time horizon, \(T\), the control logic, \(\upsilon_T\), described in Proposition \ref{prop:mainres5} can be obtained by solving a set of convex optimization problems. 
    \begin{prop}
    \label{prop:mainres6}
        Suppose that the following conditions hold:
        \begin{enumerate}[label=F\arabic*),leftmargin=*]
            \item\label{condn:f1-2} Conditions \ref{condn:e1} and \ref{condn:e2} hold.
            \item\label{condn:f3} The \(\bigl(s_i\bigr)_{i=1}^{N}\)-sparse solutions \(\biggl(\bigl(\overline{u}_i(t)\bigr)_{t=0}^{T-1}\biggr)_{i=1}^{N}\) to the optimization problem \eqref{e:optprob1} that satisfy condition \ref{condn:e4} are the unique solutions to \eqref{e:optprob1}. 
            \item\label{condn:f4} For each plant \(i=1,2,\ldots,N\), the matrix, \(\Phi_i\), satisfies the restricted isometry property of order \(2s_i\) with \(\delta_{2s_i}<\sqrt{2}-1\).
        \end{enumerate}
        Then for each plant \(i=1,2,\ldots,N\), \(\bigl(\overline{u}_i(t)\bigr)_{t=0}^{T-1}\) is the solution to the convex optimization problem
        \begin{align}
        \label{e:optprob2}
            \underset{\bigl(u_i(t)\bigr)_{t=0}^{T-1}\in U_{\xi_i}}\minimize&\:\norm{\pmat{u_i(0)\\u_i(1)\\\vdots\\u_i(T-1)}}_{1}.
        \end{align}
    \end{prop}
    
    Proposition \ref{prop:mainres6} provides sufficient conditions on the matrices, \(\Phi_i\), \(i=1,2,\ldots,N\) under which the unique solutions, \(\biggl(\bigl(\overline{u}_i(t)\bigr)_{t=0}^{T-1}\biggr)_{i=1}^{N}\), to the optimization problem \eqref{e:optprob1} that together satisfy condition \ref{condn:e4} can be obtained by solving the optimization problem \eqref{e:optprob2} for each plant \(i=1,2,\ldots,N\). The optimization problem \eqref{e:optprob2} is convex by convexity of the objective function and the set of decision variables, see proof of Proposition \ref{prop:mainres5} for details. Notice that uniqueness of \(\bigl(\overline{u}_i(t)\bigr)_{t=0}^{T-1}\) is governed by \(\Phi_i\) and \(\xi_i\), \(i=1,2,\ldots,N\). Further, the properties of \(\Phi_i\), by its definition, rely on the plant dynamics, \((A_i,b_i)\), \(i=1,2,\ldots,N\), capacity of the communication network, \(M\) and the time horizon, \(T\). We present a proof of Proposition \ref{prop:mainres6} in \S\ref{s:proofs}.
 
   To summarize, in this section we discussed methods and conditions for solving the feasibility problem \eqref{e:feasprob1} that is at the heart of our design of scheduling and control logic in Algorithm \ref{algo:sched-con_design}.
	\begin{remark}
	\label{rem:hands-off}
    \rm{
        Consider a linear system \(x(t+1)=Ax(t)+bu(t)\), \(t=0,1,\ldots,T-1\). A sparsest control sequence, \(\bigl(u(t)\bigr)_{t=0}^{T-1}\), that steers a given \(x(0)=\xi\) to \(x(T) = 0\) is called a \emph{maximum hands-off control} sequence. Design of such sequences has attracted a considerable research attention in the recent past, see e.g., \cite{Nagahara2016a,Nagahara2016b}. Our design of a scheduling and control logic, \((\gamma_T,\upsilon_T)\), presented in this paper is similar in spirit to maximum hands-off control. Indeed, instead of minimizing the \(\ell_0\)-norm of control sequence for a specific plant, i.e., \(\pmat{u_i(0)\\\vdots\\u_i(T-1)}\), we minimize the \(\ell_0\)-norm of the \(t\)-th element of the control sequence of each plant, i.e., \(\pmat{u_1(t)\\\vdots\\u_N(t)}\). As discussed in \S\ref{s:proofs}, our proofs of Propositions \ref{prop:mainres3}-\ref{prop:mainres6} rely on the theory of maximum hands-off control presented in \cite{Nagahara2016b}.
        }
	\end{remark}

\section{Numerical example}
\label{s:numex}
\subsection{The NCS}
\label{ss:ncs}
    We consider an NCS with \(N=100\) discrete-time linear plants and a shared communication network of limited capacity \(M=10\). The plant dynamics, \((A_i,b_i)\), \(i=1,2,\ldots,N\) are chosen as follows:
    \begin{itemize}[label=\(\circ\), leftmargin=*]
        \item Elements of \(A_i\in\R^{2\times 2}\), \(i=1,2,\ldots,50\) and \(A_i\in\R^{3\times 3}\), \(i=51,52,\ldots,100\) are selected from the interval \([-2,+2]\) uniformly at random.
        \item It is ensured that \(A_i\), \(i=1,2,\ldots,N\) are Schur unstable. 
        \item Elements of \(b_i\in\R^2\), \(i=1,2,\ldots,50\) and \(b_i\in\R^3\), \(i=51,52,\ldots,100\) are selected from the interval \([-2,+2]\) uniformly at random.
        \item It is ensured that \((A_i,b_i)\), \(i=1,2,\ldots,N\) are reachable.
    \end{itemize}
    
    We let the time horizon \(T=50\) units of time. We choose the initial conditions, \(x_i(0) = \xi_i\in\R^2\), \(i=1,2,\ldots,50\) and \(x_i(0) = \xi\in\R^3\), \(i=51,52,\ldots,100\) from the intervals \([-1,+1]^2\) and \([-1,+1]^3\), respectively, uniformly at random.
\subsection{Scheduling and control logic, \((\gamma_T,\upsilon_T)\)}
\label{ss:sched-con_design}
    We employ Algorithm \ref{algo:sched-con_design} to design a control logic, \(\upsilon_T\). 
    
    We have that each plant \(i=1,2,\ldots,N\) is reachable. Thus, condition \ref{condn:c1} in Proposition \ref{prop:mainres3} holds. Let \(\P_1 = \{1,2,\ldots,10\}\), \(\P_2 = \{11,12,\ldots,20\}\), \(\P_3 = \{21,22,\ldots,30\}\), \(\P_4 = \{31,32,\ldots,40\}\), \(\P_5 = \{41,42,\ldots,50\}\), \(\P_6 = \{51,52,\ldots,60\}\), \(\P_7 = \{61,62,\ldots,70\}\), \(\P_8 = \{71,72,\ldots,80\}\), \(\P_9 = \{81,82,\ldots,90\}\), \(\P_{10} = \{91,92,\ldots,100\}\), and 
    \(\hat{d}_1=\hat{d}_2=\hat{d}_3=\hat{d}_4=\hat{d}_5=3\), \(\hat{d}_6=\hat{d}_7=\hat{d}_8=\hat{d}_9=\hat{d}_{10}=4\). It follows that \(\abs{\P_j} = 10 = M\), \(j=1,2,\ldots,10\), 
    \(\P_\ell\cap\P_m = \emptyset\) for all \(\ell,m=1,2,\ldots,10\), \(\ell\neq m\), \(\displaystyle{\bigcup_{j=1}^{10}\P_j} = \{1,2,\ldots,N\}\), \(\hat{d}_j>d_i\) for all \(i\in\P_j\), \(j=1,2,\ldots,10\), and \(\displaystyle{\sum_{j=1}^{10}}\hat{d}_j = 35 < 50\). Thus, condition \ref{condn:c2} in Proposition \ref{prop:mainres3} holds. Algorithm \ref{algo:input_design1} is applicable for solving the feasibility problem \eqref{e:feasprob1} and we utilize it.
    
    The resulting control logic, \(\upsilon_T\) and the corresponding scheduling logic, \(\gamma_T\), are plotted in Figures \ref{fig:con_logic} and \ref{fig:sched_logic}, respectively.
    \begin{figure}[htbp]
        \includegraphics{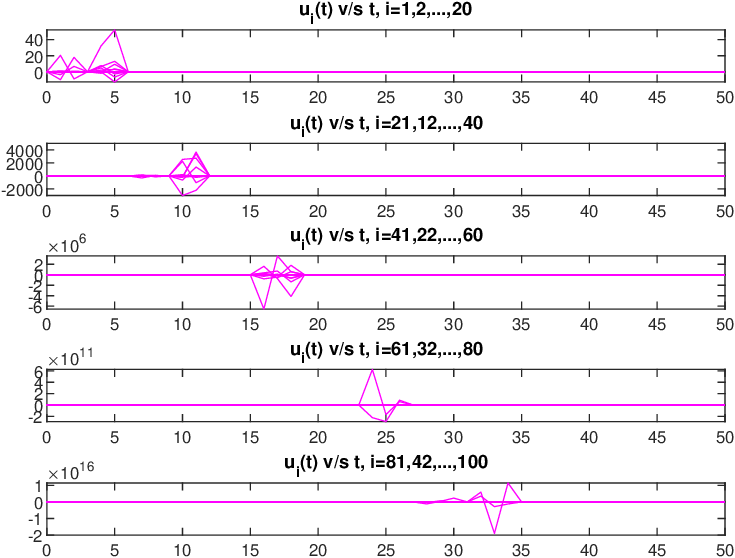}
        \caption{Control logic, \(\upsilon_T\)}\label{fig:con_logic}
    \end{figure}
    \begin{figure}[htbp]
        \includegraphics{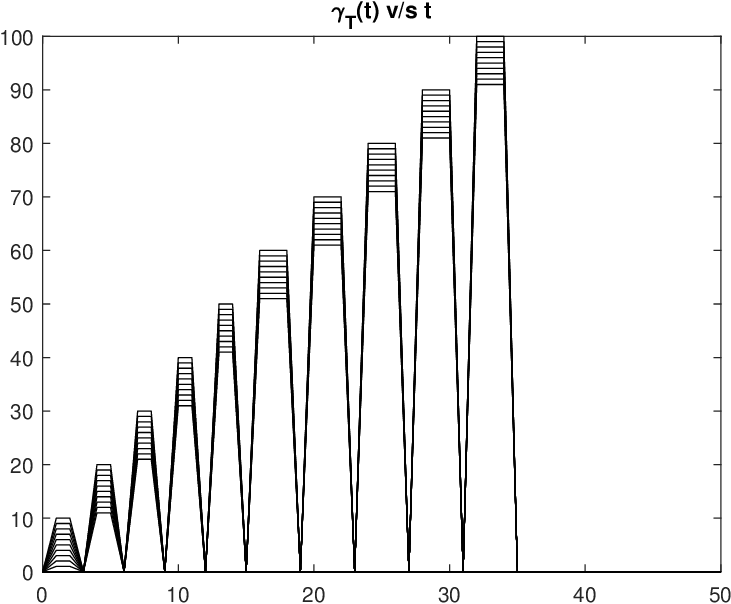}
        \caption{Scheduling logic, \(\gamma_T\)}\label{fig:sched_logic}
    \end{figure}
\subsection{State trajectories of the plants}
\label{ss:trajectory}
    We plot \(\bigl(\norm{x_i(t)}_{2}\bigr)_{t=0}^{50}\), \(i=1,2,\ldots,N\) in Figure \ref{fig:xplot}. Grouping of  plants is done to convey the effect of the location of non-zero control inputs clearly. For each plant \(i=1,2,\ldots,N\), it is observed that \(\xi_i\) is steered to \(0_{d_i}\) in \(T\) units of time.
\begin{figure}[htbp]
        \includegraphics{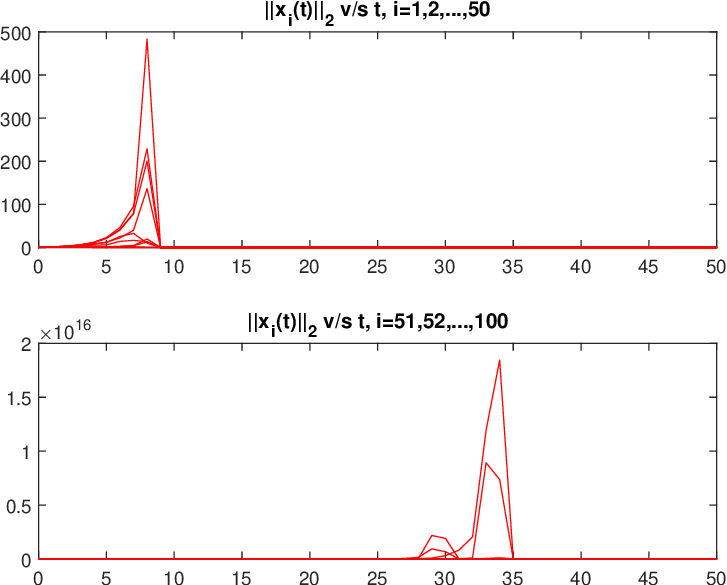}
        \caption{State trajectories, \(\bigl(\norm{x_i(t)}_2\bigr)_{t=0}^{50}\)}\label{fig:xplot}
    \end{figure}

    \begin{remark}
    \label{rem:numex1}
    \rm{
        As an alternative to Algorithm \ref{algo:input_design1}, one may employ Algorithm \ref{algo:input_design2} to solve the feasibility problem \eqref{e:feasprob1}. With the choice \(p=10=M\), \(\Nset_j=\P_j\), \(j=1,2,\ldots,10\), and \(\hat{d}_1=\hat{d}_2=\hat{d}_3=\hat{d}_4=\hat{d}_5=3\), \(\hat{d}_6=\hat{d}_7=\hat{d}_8=\hat{d}_9=\hat{d}_{10}=4\), we have that condition \ref{condn:d2} in Proposition \ref{prop:mainres4} holds. Application of Proposition \ref{prop:mainres5}-\ref{prop:mainres6} towards solving the feasibility problem \eqref{e:feasprob1} would, however, require determining whether or not for each plant \(i=1,2,\ldots,N\), the optimization problem \eqref{e:optprob1} admits a unique \(s_i\)-sparse solution. In general, due to the non-convexity of \(\ell_0\)-norm, solving the optimization problem \eqref{e:optprob1} and asserting uniqueness of its solution is a numerically difficult problem.
        }
    \end{remark}
\section{Concluding remarks}
\label{s:concln}
    In this paper we studied the design of scheduling and control logic for NCSs whose shared communication network has a limited communication capacity. Given a set of discrete-time linear plant dynamics, a set of non-zero initial conditions for the plants, the capacity of the communication network and a time horizon, a scheduling and control logic obtained from our algorithm ensures that the initial conditions are steered to zero in the given time horizon for all the plants. We employ feasibility problem with sparsity constraint (\(\ell_0\)-constraint) to compute the control logic and the corresponding scheduling logic is designed by allocating the network to the plants with non-zero control inputs. We also present a set of sufficient conditions under which our algorithm can be implemented numerically. 
    
    We identify the extension of our techniques to the design of scheduling and control logics when the communication networks are also prone to uncertainties like delays, data losses, etc. as a topic for further work. It is currently under investigation and will be reported elsewhere.
\section{Proofs of our results}
\label{s:proofs}
    \begin{proof}[Proof of Proposition \ref{prop:mainres0a}]
        Fix a scheduling and control logic, \((\gamma_T,\u_T)\), that satisfies \(\gamma_T(t)=\emptyset\) and \(\u_T(t) = 0\in\R^N\) for all \(t=0,1,\ldots,T-1\). Then for each plant \(i=1,2,\ldots,N\), we have 
        \[
            x_i(T) = A_i^{T}x_i(0) = A_i^T\xi_i = A_i^{T-\tau_i+\tau_i}\xi_i = A_i^{T-\tau_i}\bigl(A_i^{\tau_i}\xi_i\bigr) = A_i^{T-\tau_i}0_{d_i} = 0_{d_i}.
        \]
        The assertion of Proposition \ref{prop:mainres0a} is immediate.
     \end{proof}
     \begin{proof}[Proof of Proposition \ref{prop:mainres0b}]
        Fix a scheduling and control logic, \((\gamma_T,\u_T)\), that satisfy 
         \begin{align*}
            \gamma_T(t)&\subseteq\mathcal{N},\:\:t=0,1,\ldots,T-1,\:\text{and}\\
            \nu_T(t) &= \pmat{u_1(t)\\u_2(t)\\\vdots\\u_N(t)}\:\:\text{with}\:u_i(t)=0,\:\text{if}\:i\notin\mathcal{N}, \:u_i(t) = \overline{u}_i(t),\:\text{if}\:i\in\mathcal{N},\:t=0,1,\ldots,T-1,
        \end{align*}
        where \(\biggl(\overline{u}_i(t)\biggr)_{t=0}^{T-1}\in U_{\xi_i}\).
        
        Then for each plant \(i\in\mathcal{N}\), \(x_i(T) = 0_{d_i}\) under \(\biggl(\overline{u}_i(t)\biggr)_{t=0}^{T-1}\), and for each plant \(i\notin\mathcal{N}\), \(x_i(T) = A_{i}^{T-\tau_i+\tau_i}\xi_i = 0_{d_i}\).
        
        The assertion of Proposition \ref{prop:mainres0b} follows at once.
     \end{proof}
     \begin{proof}[Proof of Proposition \ref{prop:mainres0c}]
        Assume \(T < \lceil\frac{N}{M}\rceil\). Thus, there is at least one plant \(\overline{i}\in\{1,2,\ldots,N\}\) such that \(\overline{i}\notin\gamma_T(t)\) for all \(t=0,1,\ldots,T-1\). Since there is no \(\tau\in\{1,2,\ldots,T\}\) such that \(A_{\overline{i}}^{\tau}\xi_i=0_{d_i}\), we have that \(x_{\overline{i}}(T)\neq O_{d_i}\) under \(\gamma_T\).
        
        An application of contrapositive logic leads us to the assertion of Proposition \ref{prop:mainres0c}.
     \end{proof}
    \begin{proof}[Proof of Proposition \ref{prop:mainres1}]
        Let \(\upsilon_T\) be a solution to the feasibility problem \eqref{e:feasprob1}. 
        
        Firstly, we have that \(\gamma_T(t)\) is the set containing the elements of \(\{1,2,\ldots,N\}\) such that \({u}_i(t) \neq 0\), \(t=0,1,\ldots,T-1\). By construction of \(\upsilon_T\),
        \(\gamma_T(t)\) contains at most \(M\) distinct elements of the set \(\{1,2,\ldots,N\}\), \(t=0,1,\ldots,T-1\). It follows that \(\gamma_T\) obtained from Algorithm \ref{algo:sched-con_design} is a valid scheduling logic.
        
        Secondly, by construction of \(\gamma_T\), we have that \({u}_i(t) = 0\) whenever \(i\) is not an element of \(\gamma_T(t)\), \(t=0,1,\ldots,T-1\). Thus, \(\upsilon_T\) obtained from Algorithm \ref{algo:sched-con_design} is a valid control logic.
        
        Finally, by properties of \(\upsilon_T\), we have that for each plant \(i=1,2,\ldots,N\), \(\bigl({u}_i(t)\bigr)_{t=0}^{T-1}\) steers \(x_i(0) = \xi_i\) to \(x_i(T) = 0_{d_i}\) while obeying the state evoluation \(x_i(t+1) = A_ix_i(t) + b_i{u}_i(t)\), \(t=0,1,\ldots,T-1\). 
        
        The assertion of Proposition \ref{prop:mainres1} follows at once.
    \end{proof}
     The following lemma will be useful in our proofs of Propositions \ref{prop:mainres3} and \ref{prop:mainres4}. 
    \begin{lemma}
    \label{lem:auxres1}
        Fix plant \(i\in\{1,2,\ldots,N\}\). Suppose that the plant \(i\) is reachable. Let \(x_i(0) = \xi_i\in\R^{d_i}\) and \(\hat{d}\in\N\) satisfying \(\hat{d}>d_i\) be given. Then there exists a sequence of control inputs \(\bigl(u_i(t)\bigr)_{t=0}^{\hat{d}-1}\) under which \(x_i(\hat{d})=0_{d_i}\).
    \end{lemma}
    
    Lemma \ref{lem:auxres1} follows from \cite[Lemma 1]{Nagahara2016b}, In particular, the choice of \(\bigl(u_i(t)\bigr)_{t=0}^{\hat{d}-1}\) is given by
    \begin{align*}
        \pmat{u_i(0)\\u_i(1)\\\vdots\\u_i(\hat{d}-1)} = \pmat{0_{\hat{d}-d_i}\\-\Psi_i^{-1}A_i^{\hat{d}}\xi_i}.
    \end{align*}

     \begin{proof}[Proof of Proposition \ref{prop:mainres3}]
      Let \(\upsilon_T\) be a control logic obtained from Algorithm \ref{algo:input_design1}. We will show that \(\upsilon_T\) is a solution to the feasibility problem \eqref{e:feasprob1}.
     
     First, in view of the hypotheses \(\P_{\ell}\cap\P_{m}=\emptyset\) for all \(\ell,m = 1,2,\ldots,\lceil\frac{N}{M}\rceil\), \(\ell\neq m\) and \(\displaystyle{\bigcup_{j=1}^{\lceil\frac{N}{M}\rceil}}\P_j = \{1,2,\ldots,N\}\), and the construction of \({u}_i\), we have that \({u}_i(t)\) is well-defined for all \(t=0,1,\ldots,T-1\) and all \(i=1,2,\ldots,N\).
     
     Second, fix a plant \(i\in\{1,2,\ldots,N\}\). The sequence of control inputs \(\bigl({u}_i(t)\bigr)_{t=0}^{T-1}\) chosen as \(\displaystyle{\pmat{{u}_i(0)\\\vdots\\{u}_i(\tilde{d}_j-1)}} =0_{\tilde{d}_j}\), \(\displaystyle{\pmat{u_i\bigl(\tilde{d}_j+0\bigr)\\\vdots\\u_i\bigl(\tilde{d}_j+\hat{d}_j-1\bigr)} = \pmat{0_{\hat{d}_j-d_i}\\-\Psi_i^{-1}A_i^{\hat{d}_j}A_i^{\tilde{d}_j}{\xi}_i}}\), and \(\pmat{{u}_i(\tilde{d}_j+\hat{d}_j)\\\vdots\\{u}_i(T-1)} = 0_{T-\tilde{d}_j-\hat{d}_j}\).
     steers the state \(x_i(0) = \xi_i\) to \(x_i(T)=0_{d_i}\) under the dynamics \(x_i(t+1) = A_i x_i(t) + b_i{u}_i(t)\), \(t=0,1,\ldots,T-1\). Indeed, 
     \begin{itemize}[label=\(\circ\),leftmargin=*]
        \item with \(\bigl({u}_i(t)\bigr)_{t=0}^{\tilde{d}_j-1}\), the state \(x_i(0)=\xi_i\) is steered to \(x_i(\tilde{d}_j)=A_{\tilde{d}_j}\xi_i\),
        \item in view of Lemma \ref{lem:auxres1}, with \(\bigl({u}_i(t)\bigr)_{t=\tilde{d}_j+0}^{\tilde{d}_j+\hat{d}_j-1}\), the state \(x_i(\tilde{d}_j)=A_i^{\tilde{d}_j}\xi_i\) is steered to \(x_i(\tilde{d}_j+\hat{d}_j)=0_{d_i}\), and
        \item with \(\bigl({u}_i(t)\bigr)_{t=\tilde{d}_j+\hat{d}_j}^{T-1}\), the state \(x_i(\tilde{d}_j+\hat{d}_j)=0_{d_i}\) continues to be \(0_{d_i}\) for all \(t= \tilde{d}_j+\hat{d}_j+1,\ldots,T-1\).
     \end{itemize}
     Also, from Lemma \ref{lem:auxres1} and the hypothesis that \(\displaystyle{\sum_{j=1}^{\lceil\frac{N}{M}\rceil}}\hat{d}_j\leq T\), it follows that \(x_i(0)=\xi_i\) is steered to \(x_i(T)=0_{d_i}\) for all plants \(i=1,2,\ldots,N\) in the time horizon \(\{0,1,\ldots,T-1\}\).
     
     Third, by the construction of \(\upsilon_T\), for any \(t\in\{0,1,\ldots,T-1\}\), \({u}_i(t)\) is possibly non-zero only for \(i\in\{1,2,\ldots,N\}\cap\P_j\) for exactly one \(j\in\{1,2,\ldots,\bigl\lceil\frac{N}{M}\bigr\rceil\}\). Indeed, we have different values of \(\tilde{d}_j\) for each \(j=1,2,\ldots,\bigl\lceil\frac{N}{M}\bigr\rceil\). Since \(\abs{\P_j}\leq M\) for all \(j=1,2,\ldots,\bigl\lceil\frac{N}{M}\bigr\rceil\), it follows that \(\displaystyle{\norm{\upsilon_T(t)}_{0}}\leq M\), \(t=0,1,\ldots,T-1\). 
     
     Consequently, \(\upsilon_T\) is a solution to the feasibility problem \eqref{e:feasprob1}.
     \end{proof}

    \begin{proof}[Proof of Proposition \ref{prop:mainres4}]
        Let \(\upsilon_T\) be a control logic obtained from Algorithm \ref{algo:input_design2}. We will show that \(\upsilon_T\) is a solution to the feasibility problem \eqref{e:feasprob1}.
    
    First, in view of the hypotheses \(\Nset_m\cap\Nset_n=\emptyset\) for all \(m,n=1,2,\ldots,p\), \(m\neq n\) and \(\displaystyle{\bigcup_{j=1}^{p}\Nset_j=\{1,2,\ldots,N\}}\), and construction of \({u}_i\) in Algorithm \ref{algo:input_design2}, we have that \({u}_i(t)\) is well-defined for all \(t=0,1,\ldots,T-1\) and all \(i=1,2,\ldots,N\).
    
    Second, for all \(i=1,2,\ldots,N\), the sequence of control inputs \(\bigl({u}_i(t)\bigr)_{t=0}^{T-1}\) chosen as \( \displaystyle{\pmat{{u}_{\Nset_j(k)}(0)\\\vdots\\{u}_{\Nset_j(k)}(\tilde{d}_i-1)}} =0_{\tilde{d}_i}\), 	\(\pmat{u_{\Nset_j(k)}\bigl(\tilde{d}_{\Nset_j(k)}+0\bigr)\\\vdots\\u_{\Nset_j(k)}\bigl(\tilde{d}_{\Nset_j(k)}+\hat{d}_{\Nset_j(k)}-1\bigr)} = \pmat{0_{\hat{d}_{\Nset_j(k)}-d_{\Nset_j(k)}}\\-\Psi_{\Nset_j(k)}^{-1}A_{\Nset_j(k)}^{\hat{d}_{\Nset_j(k)}}A_{\Nset_j(k)}^{\tilde{d}_{\Nset_j(k)}}{\xi}_{\Nset_j(k)}}\) and \( \pmat{\overline{u}_{\Nset_j(k)}(\tilde{d}_{\Nset_j(k)}+\hat{d}_{\Nset_j(k)})\\\vdots\\\overline{u}_{\Nset_j(k)}(T-1)} = 0_{T-\tilde{d}_{\Nset_j(k)}-\hat{d}_{\Nset_j(k)}}\)
     steers the state \(x_i(0) = \xi_i\) to \(x_i(T)=0_{d_i}\) under the dynamics \(x_i(t+1) = A_i x_i(t)+b_i u_i(t)\), \(t=0,1,\ldots,T-1\). A similar explanation as provided in our proof of Proposition \ref{prop:mainres3} holds. Also, in view of Lemma \ref{lem:auxres1} and the hypothesis that \(\displaystyle{\sum_{i\in\Nset_j}}\hat{d}_i\leq T\), it follows that \(x_i(0) = \xi_i\) is steered to \(x_i(T) = 0_{d_i}\) for all plants \(i=1,2,\ldots,N\) in the time horizon \(\{0,1,\ldots,T-1\}\). 
     
     Third, since \(p\leq M\) and by construction of \(\upsilon_T\), we have that for any \(t\in\{0,1,\ldots,T-1\}\), \({u}_i(t)\) is possibly non-zero only for \(p\)-many \(i\in\{1,2,\ldots,N\}\). It follows that \(\norm{\upsilon_T}_0\leq M\), \(t=0,1,\ldots,T-1\).
     
     Consequently, \(\upsilon_T\) is a solution to the feasibility problem \eqref{e:feasprob1}.
   \end{proof}
    \begin{proof}[Proof of Proposition \ref{prop:mainres5}]
       We first show that the optimization problem \eqref{e:optprob2} is convex. Notice that the problem can be rewritten as
       \begin{align*}
            \underset{\bigl(u_i(t)\bigr)_{t=0}^{T-1}\in\R^T}\minimize&\:\norm{\pmat{u_i(0)\\u_i(1)\\\vdots\\u_i(T-1)}}_{1}\\
            \sbjto&\:
            \begin{cases}
                x_i(t+1) = A_ix_i(t)+b_iu_i(t),\:t=0,1,\ldots,T-1,\\
                x_i(0) = \xi_i,\:x_i(T) = 0_{d_i},
            \end{cases}
       \end{align*}
       which is equivalent to
       \begin{align}
       \label{e:optproba}
            \underset{\bigl(u_i(t)\bigr)_{t=0}^{T-1}\in\R^T}\minimize&\:\norm{\pmat{u_i(0)\\u_i(1)\\\vdots\\u_i(T-1)}}_{1}\\
            \sbjto&\:\Phi_i\pmat{u_i(0)\\u_i(1)\\\vdots\\u_i(T-1)} = -A_i^T\xi_i.
       \end{align}
       
       We have that for the optimization problem \eqref{e:optproba}, the objective function is convex as \(\ell_1\)-norm is a convex function \cite[Chapter 2]{Nagahara_book} and the constraint function is convex in view of linearity.
       
       It remains to show that if \(\biggl(\bigl(\overline{u}_i(t)\bigr)_{t=0}^{T-1}\biggr)_{i=1}^{N}\) that satisfy condition \ref{condn:e4} are unique and \(\Phi_i\), \(i=1,2,\ldots,N\) satisfy restricted isometry property of order \(2s_i\) with \(\delta_{2s_i}<\sqrt{2}-1\), \(i=1,2,\ldots,N\), then \(\biggl(\bigl(\overline{u}_i(t)\bigr)_{t=0}^{T-1}\biggr)_{i=1}^{N}\) can be obtained by solving the optimization problem \eqref{e:optprob2} for all \(i=1,2,\ldots,N\). 
       
       In view of \cite[Theorem 1]{Nagahara2016b} we have that if \(\bigl(\tilde{u}_i(t)\bigr)_{t=0}^{T-1}\) is a unique solution to \eqref{e:optprob1}, then it is equivalent to the solution to \eqref{e:optprob2} under conditions \ref{condn:f1-2}-\ref{condn:f4}. The assertion of Proposition \ref{prop:mainres6} follows at once.
    \end{proof}
    \begin{proof}[Proof of Proposition \ref{prop:mainres6}]
        Let \(x_i(0)=\xi_i\), \(i=1,2,\ldots,N\) be given. Since all plants \(i=1,2,\ldots,N\) are reachable and \(T>d_i\) for all \(i=1,2,\ldots,N\), it follows from Lemma \ref{lem:auxres1} that there exist \(\biggl(\bigl(\overline{u}_i(t)\bigr)_{t=0}^{T-1}\biggr)_{i=1}^{N}\) under which \(x_i(T) = 0_{d_i}\), \(i=1,2,\ldots,N\). 
        
        Let \(\tilde{u}_i(t) = u_i(t)\) for all \(t=0,1,\ldots,T-1\) and \(i=1,2,\ldots,N\). Let \(\biggl(\bigl(\overline{u}_i(t)\bigr)_{t=0}^{T-1}\biggr)_{i=1}^{N}\) be \(\bigl(s_i\bigr)_{i=1}^{N}\)-sparse and satisfy condition \ref{condn:e4}. We need to show that the control logic, \(\upsilon_T(t) = \pmat{\overline{u}_1(t)\\\overline{u}_2(t)\\\vdots\\\overline{u}_N(t)}\), \(t=0,1,\ldots,T-1\) is a solution to the feasibility problem \eqref{e:feasprob1}.
        
        We first note that since \(\bigl(\overline{u}_i(t)\bigr)_{t=0}^{T-1}\in U_{\xi_i}\), \(i=1,2,\ldots,N\), \(\bigl(\overline{u}_i(t)\bigr)_{t=0}^{T-1}\) steers \(x_i(0) = \xi_i\) to \(x_i(T)=0_{d_i}\), \(i=1,2,\ldots,N\). Since \(\displaystyle{\sum_{i\in\mathcal{R}_j}s_i\leq T}\), for each plant \(i\), \(\bigl(\overline{u}_i(t)\bigr)_{t=0}^{T-1}\) can be accommodated in the time horizon, \(T\). 
        
        We next note that for each \(j=1,2,\ldots,M\) and \(t=0,1,\ldots,T-1\), \(\overline{u}_i(t)\neq 0\) either for no or for exactly one \(i\in\mathcal{R}_j\). Thus, for at most \(M\) elements in \(\{1,2,\ldots,N\}\), \(\overline{u}_i(t) \neq 0\) for every \(t=0,1,\ldots,T-1\). It follows that \(\norm{\upsilon_T(t)}_{0}\leq M\), \(t=0,1,\ldots,T-1\).
    \end{proof}


\begin{thebibliography}{10}

\bibitem{Al-Areqi'15}
{\sc S.~Al-Areqi, D.~G\"orges, and S.~Liu}, {\em Event-based control and
  scheduling codesign: stochastic and robust approaches}, IEEE Transactions on
  Automatic Control, 60 (2015), pp.~1291--1303.

\bibitem{Dai2010}
{\sc S.-L. Dai, H.~Lin, and S.~S. Ge}, {\em Scheduling-and-control codesign for
  a collection of networked control systems with uncertain delays}, IEEE
  Transactions on Control Systems Technology, 18 (2010), pp.~66--78.

\bibitem{Gatsis2016}
{\sc K.~Gatsis, A.~Ribeiro, and G.~Pappas}, {\em Control-aware random access
  communication}.
\newblock Proceedings of the 7th ACM/IEEE International Conference on
  Cyber-Physical Systems (ICCPS), 2016, Vienna, Austria, pp. 1-9.

\bibitem{Hristu2001}
{\sc D.~Hristu-Varsakelis}, {\em Feedback control systems as users of a shared
  network: Communication sequences that guarantee stability}.
\newblock Proceedings of the 40th IEEE Conference on Decision and Control,
  2001, Orlando, Florida, USA, pp. 3631-3636.

\bibitem{Hristu2008}
{\sc D.~Hristu-Varsakelis}, {\em Short-period communication and the role of
  zero-order holding in networked control systems}, IEEE Transactions on
  Automatic Control, 53 (2008), pp.~1285--1290.

\bibitem{Hristu_Zhang2008}
{\sc D.~Hristu-Varsakelis and L.~Zhang}, {\em L{QG} control of networked
  control systems with access constraints and delays}, International Journal of
  Control, 81 (2008), pp.~1266--1280.

\bibitem{Ikeda2022}
{\sc T.~Ikeda and K.~Kashima}, {\em Sparse control node scheduling in networked
  systems based on approximate controllability metrics}, IEEE Trans. Control
  Netw. Syst., 9 (2022), pp.~1166--1177.

\bibitem{def}
{\sc A.~Kundu}, {\em Scheduling networked control systems under jamming
  attacks}, Proceedings of the 59th IEEE Conference on Decision and Control,
  (2020), p.~3298–3303.

\bibitem{abc}
{\sc A.~Kundu and D.~E. Quevedo}, {\em Stabilizing scheduling policies for
  networked control systems}, IEEE Transactions on Control of Network Systems,
  7 (2020), pp.~163--175.

\bibitem{ghi}
\leavevmode\vrule height 2pt depth -1.6pt width 23pt, {\em On periodic
  scheduling and control for networked systems under random data loss}, IEEE
  Transactions on Control of Network Systems, 8 (2021), pp.~1788--1798.

\bibitem{Liberzon}
{\sc D.~Liberzon}, {\em Switching in systems and control}, Systems \& Control:
  Foundations \& Applications, Birkh\"auser Boston Inc., Boston, MA, 2003.

\bibitem{Lin2005}
{\sc H.~Lin, G.~Zhai, L.~Fang, and P.~J. Antsaklis}, {\em Stability and h-inf
  performance preserving scheduling policy for networked control systems}.
\newblock Proc. of the 16th IFAC World Congress, 2005, Prague, Czech Republic.

\bibitem{Quevedo2014}
{\sc M.~Lje\v{s}njanin, D.~E. Quevedo, and D.~Ne\v{s}i\'{c}}, {\em Packetized
  {MPC} with dynamic scheduling constraints and bounded packet dropouts},
  Automatica J. IFAC, 50 (2014), pp.~784--797.

\bibitem{Ma2019}
{\sc Y.~Ma, J.~Guo, Y.~Wang, A.~Chakrabarty, H.~Ahn, P.~Orlik, and C.~Lu}, {\em
  Optimal dynamic scheduling of wireless networked control systems}.
\newblock Proceedings of the 10th ACM/IEEE International Conference on
  Cyber-Physical Systems, 2019, Montreal, Canada, pp. 77-86.

\bibitem{Nagahara_book}
{\sc M.~Nagahara}, {\em Sparsity Methods for Systems and Control}, Now
  Publishers, 2020.

\bibitem{Nagahara2016b}
{\sc M.~Nagahara, J.~Ostergaard, and D.~E. Quevedo}, {\em Discrete-time
  hands-off control by sparse optimization}, EURASIP Journal on Advances in
  Signal Processing, 76 (2016).

\bibitem{Nagahara2016a}
{\sc M.~Nagahara, D.~E. Quevedo, and D.~Ne\v{s}i\'{c}}, {\em Maximum hands-off
  control: a paradigm of control effort minimization}, IEEE Trans. Automat.
  Control, 61 (2016), pp.~735--747.

\bibitem{Peters'16}
{\sc E.~G.~W. Peters, D.~E. Quevedo, and M.~Fu}, {\em Controller and scheduler
  codesign for feedback control over {IEEE} 802.15.4 networks}, IEEE
  Transactions on Control Systems Technology, 24 (2016), pp.~2016--2030.

\bibitem{Rehbinder2004}
{\sc H.~Rehbinder and M.~Sanfridson}, {\em Scheduling of a limited
  communication channel for optimal control}, Automatica, 40 (2004),
  pp.~491–--500.

\bibitem{Saha2015}
{\sc I.~Saha, S.~Baruah, and R.~Majumdar}, {\em Dynamic scheduling for
  networked control systems}.
\newblock Proceedings of the 18th ACM International Conference on Hybrid
  Systems: Computation \& Control, 2015, Seattle, Washington, pp. 98-107.

\bibitem{Zhang2006}
{\sc L.~Zhang and D.~Hristu-Varsakelis}, {\em Communication and control
  co-design for networked control systems}, Automatica J. IFAC, 42 (2006),
  pp.~953--958.

\end{thebibliography}

\end{document}